\documentclass[11pt,a4paper]{article} 

\usepackage{amsmath,amssymb}
\usepackage{comment}
\usepackage{graphicx}
\usepackage[lined]{algorithm2e} 

\newlength{\wider}
\def\squareforqed{\hbox{\rlap{$\sqcap$}$\sqcup$}}
\def\qed{\ifmmode\squareforqed\else{\unskip\nobreak\hfil
\penalty50\hskip1em\null\nobreak\hfil\squareforqed
\parfillskip=0pt\finalhyphendemerits=0\endgraf}\fi}
\newtheorem{theorem}{Theorem}
\newtheorem{lemma}[theorem]{Lemma}

\newtheorem{definition}[theorem]{Definition}

\newenvironment{proof}{\begin{trivlist}\item[]{\flushleft\bf Proof }}
{\qed\end{trivlist}}

\title{Finding Optimal Flows Efficiently}
\author{Mehdi Mhalla\footnote{LIG, University of Grenoble, France, mehdi.mhalla@imag.fr} ,
Simon Perdrix\footnote{Oxford University Computing Laboratory, simon.perdrix@comlab.ox.ac.uk}}
\date{}
\begin{document}
\maketitle
\begin{abstract} \footnotesize 
Among the models of quantum computation, the One-way Quantum Computer
\cite{BR,RB} is one of the most promising proposals of physical
realization \cite{WRRS}, and opens new perspectives for
parallelization by taking advantage of quantum entanglement
\cite{Annelham}.    
Since a one-way quantum computation is based on quantum measurement, which is a fundamentally nondeterministic evolution, a sufficient condition of global determinism has been introduced in \cite{DK}  as the existence of a \emph{causal flow} in a graph that underlies the computation. 
 A $O(n^3)$-algorithm has been introduced \cite{B07} for finding such a causal flow when the numbers of output and input vertices in the graph are equal, otherwise no polynomial time algorithm\footnote{an exponential time algorithm is proposed in \cite{B07}} was known for deciding whether a graph has a causal flow or not.
Our main contribution is to introduce  a $O(n^2)$-algorithm for finding a causal flow, if any,  whatever the numbers of input and output vertices are. This answers the open question stated by Danos and Kashefi \cite{DK} and by de Beaudrap \cite{B07}. Moreover, we prove that our algorithm produces an optimal flow (flow of minimal depth.)

Whereas the existence of a causal flow is a sufficient condition for determinism, it is not a necessary condition. A weaker version of the causal flow, called \emph{gflow} (generalized flow) has been introduced in \cite{BKMP} and has been proved to be a necessary and sufficient condition for 
a family of deterministic computations. Moreover the depth of the
quantum computation is upper bounded by the depth of the gflow.  
However, the existence of a polynomial time algorithm that finds a
gflow has been stated as  an open question in \cite{BKMP}. In this
paper we answer this positively with a polynomial time algorithm that outputs an optimal gflow of a given graph and thus finds an optimal correction strategy to the nondeterministic evolution due to measurements.

\end{abstract}

\section{Introduction}

A one-way quantum computation \cite{BR} consists in performing a
sequence of one-qubit measurements on an initial entangled quantum
state described by a graph and called graph state \cite{Graphstate}
where some vertices correspond to the input qubits of the computation,
others to the output qubits and the rest of the vertices correspond to auxiliary qubits measured during the computation. 
Since quantum measurements are nondeterministic, a one-way quantum computation  requires corrections which depend on the results of the measurements, and  which should induce a minimal depth for the computation.

 Because of these corrections, not all graph states can be used for  deterministic computation. 
The measurement calculus  \cite{DKP} is a formal framework for one-way quantum
computations, where the dependencies between
measurements and corrections are precisely identified.
Using this formalism, Danos and Kashefi  in \cite{DK} proved that a one-way quantum computation obtained by translation from a quantum circuit is such that the underlying graph satisfies a causal flow condition (see section \ref{sec:def}.)

In \cite{B07b} a polynomial time algorithm in the size of the graph has been proposed for finding a causal flow when the numbers of outputs and inputs are equal, whereas the existence of a polynomial time algorithm in the general case, has been stated as an open question. 
We propose in this paper a faster and more general algorithm for finding a causal flow, whenever the numbers of inputs and outputs are different. 

It turns out that the existence of a causal flow is not a necessary condition for determinism. A weaker flow condition, the gflow condition has indeed been introduced for characterizing uniform, stepwise and strong deterministic computation,  where the correction strategy does not depend on the measurement basis (see \cite {BKMP} for a formal definition.) 
Here, we  introduce a polynomial time algorithm for finding a gflow and thus checking whether a graph allows a uniform deterministic computation, which gives substantially more relevance to the notion of gflow introduced in \cite{BKMP}.

We also prove that the algorithms proposed are optimal, which means that they give a minimal depth flow. This implies   that the gflow algorithm gives a lower bound on the complexity of a correction strategy in a measurement-based setting for quantum computation.

\section{Definitions}\label{sec:def}
A graph with input and output vertices is called an open graph, and is defined as follows: 

 \begin{definition}[Open Graph]
An open graph is a triplet $(G,I,O)$, where $G=(V,E)$ is a undirected graph, and $I, O\subseteq V$ are respectively called input and output vertices.
\end{definition}

During a one-way quantum computation all non output qubits
(represented as non output vertices in the corresponding open graph)
are measured. Since quantum measurements are nondeterministic, for
each qubit measurement,  a corrective strategy consists in acting on
some unmeasured non input qubits, depending on the classical outcome
of the measurement, in order to make the computation
deterministic. Thus, a corrective strategy induces a sequential
dependance between  measurements. As a consequence, the depth of the
quantum  computation depends on the corrective strategy.

Corrective strategies will be defined by flows on the open graphs. A
flow $(g,\prec)$ consists in a partial order $\prec$ over the vertices ($i\prec j$ if $i$ is measured before $j$)  and a function $g$ that associates with each vertex, the vertices used for correcting its measurement (all non output qubits are measured.) Input qubits  cannot be used for correction (see \cite{DKP}.)
   
Given an open graph, two kinds of flows  are considered: the causal flow and the gflow (generalized flow.) The former
has been introduced by Danos and Kashefi \cite{DK} and corresponds to
the computation strategy that consists in correcting each qubit
measurement by acting on a single neighbor of the measured qubit. 
For a given open graph, a causal flow is characterized by a function $g$ which associates with each non output vertex a  non input vertex used for the correction of its measurement. More formally:

\begin{definition}[causal flow]
$(g,\prec)$ is a causal flow of $(G,I,O)$, where $g:V(G)\setminus O \to V(G)\setminus I$ and $\prec$ is a strict partial order over $V(G)$, if and only if\\
\indent 1. $i\prec g(i)$\\
\indent 2. if  $j\in N(g(i))$ then $j=i$ or $i\prec j$, where $N(v)$ is the
neighborhood of $v$\\
\indent 3. $i\in N(g(i)).$
\end{definition}

\begin{figure}[h]\label{fig:causalflow}
   \begin{center}
\includegraphics[width=4cm]{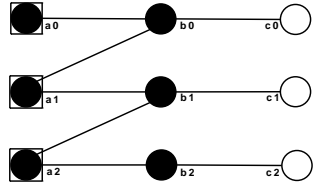}

\caption{\footnotesize{Example of open graph -- squared vertices represent inputs,
white vertices represent outputs -- which has a causal flow $(g,\prec)$,
where $g(a_i)=b_i$, $g(b_i)=c_i$ and $a_0 \prec a_1 \prec a_2 \prec
\{b_0,b_1,b_2\} \prec \{c_0,c_1,c_2\}$.}}

\end{center}
\end{figure}    

An example of causal flow is given in Figure 1. 
Notice that if the numbers of  input  and  output
vertices are the same, a causal flow can be reduced to a path cover
and then to a standard network flow  \cite{B07}. This  reduction  has
been used to define an $O(n^3)$-algorithm for finding a causal flow in the
case where the cardinalities of input and output qubits are the same \cite{B07}.

The second type of flow considered, the generalized flow, {\em
gflow}, has been introduced in \cite{BKMP} and corresponds to a more
general correction strategy that associates with each non output
vertex a set of vertices used for the corresponding correction
(instead of a single vertex.) This generalization 
not only leads to a reduction of the computational depth, but also
provides a corrective strategy to some open graphs having no causal
flow. Moreover, notice that gflow characterizes uniform, strong and stepwise deterministic computations \cite{BKMP}.

For a given open graph, a gflow $(g,\prec)$ is characterized by a function $g$ which associates with each non output vertex, a set of non input vertices used for its correction, and a strict partial order $\prec$:

\begin{definition}[gflow]\label{def:gflow}
$(g,\prec)$ is a gflow of $(G,I,O)$, where $g:V(G)\setminus O \to
\wp(V(G)\setminus I)\setminus \{\emptyset \}$ and $\prec$ is a strict
partial order over $V(G)$, if and only if\\ 
\indent 1. if $j\in g(i)$ then $i\prec j$\\
\indent 2. if  $j\in Odd(g(i))$ then $j=i$ or $i\prec j$\\
\indent 3. $i \in Odd(g(i))$\\
Where $Odd (K)=  \{u\, ,\,  |N(u)\cap K|=1 \mod 2\}$ is the odd neighborhood of $K$, i.e. the set of vertices having an odd number of neighbors in $K$.
\end{definition}

A graphical interpretation of the generalised flow is given in Figure
\ref{fig:gflow}.
 
\begin{figure}[h]\label{fig:gflow}
   \begin{center}
\includegraphics[width=6cm]{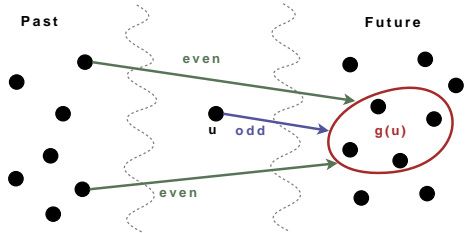}
\caption{\footnotesize{Graphical interpretation of a gflow $(g,\prec)$: for a given
vertex $u$ all the vertices larger than $u$ are in the future of $u$
since the corresponding qubits will be measured after the qubit
$u$, all others are in the past of $u$. The set $g(u)$ has to be in
the future of $u$ and such that the following parity conditions are
satisfyied: there is an odd number of edges between $g(u)$ and $u$ and
there is a even number of edges between $g(u)$ and any vertex in the past of $u$.}}
 
\end{center}
\end{figure}    

A flow $(g,\prec)$ of $(G,I,O)$ induces a partition of the vertices of
the open graph: 
\begin{definition}
For a given open graph $(G,I,O)$ and a given flow $(g,\prec)$ of $(G,I,O)$,
 let $$V^\prec_{k} = \begin{cases}\mathsf{max}_\prec(V(G))$$& \text{if
$k=0$}\\ \mathsf{max}_\prec(V(G)\setminus
V_{k-1}^\prec)& \text{if $k>0$}\end{cases}$$ where $\mathsf{max}_\prec(X) =
\{u\in X \  s.t.\  \forall v \in X, \neg (u\prec v)\}$ is the set of the
maximal elements of $X$. The \emph{depth} $d^\prec$ of the flow is the smallest
$d$ such that $V^\prec_{d+1}=\emptyset$. $(V^\prec_k)_{k=0\ldots d^\prec}$
is a partition of $V(G)$ into $d^\prec+1$ layers.
\end{definition}

A causal flow or a gflow $(g,\prec)$ of $(G,I,O)$ leads to a corrective
strategy for the corresponding one-way quantum computation, which consists in
measuring the non output qubits of each layer in parallel, from the
layer $V_{d^\prec}^\prec$ to the layer $V^\prec_1$. 
After the measurement of a layer $V^\prec_k$, with
$k>0$, corrections are realised according to the function $g$ by acting
on qubits in $\cup_{i<k} V^\prec_i$ (see \cite{BKMP} for details.)  The depth of such a one-way quantum computation is $d^\prec$.

\begin{definition}
For a given open graph $(G,I,O)$ and two given flows $(g,\prec)$ and
$(g',\prec')$ of $(G,I,O)$, $(g,\prec)$ is more delayed than
$(g',\prec')$ if 
 $\forall k$, $|\cup_{i=0\ldots k} V^\prec_k | \ge   |\cup_{i=0\ldots k} V^{\prec'}_k|$ 
and there exists a $k$ such that the inequality  is strict. \\
A causal flow (resp. gflow) $(g,\prec)$ is {\em maximally delayed}
if there exists no causal flow (resp. gflow) of the same open graph that is   more delayed. 
\end{definition}
For instance, the flow $(g,\prec)$ described in Figure
1  is a maximally delayed causal flow. However,
$(g,\prec)$ is not a maximally delayed gflow since
$(g',\prec')$ is a more delayed gflow, where
$g'(a_0)=\{b_0,b_1,b_2\}, g'(a_1)=\{b_1,b_2\}, g'(a_2)=\{b_2\},
g'(b_0)=\{c_0\},g'(b_1)=\{c_1\},g'(b_2)=\{b_2\}$, and
$\{a_0,a_1,a_2\}\prec'\{b_0,b_1,b_2\}\prec'\{c_0,c_1,c_2\}$. One can
prove that $(g',\prec')$ is a maximally delayed gflow.

\  

The following two lemmas are proved for both kinds of flows.
 
\begin{lemma}\label{lem:flowv0}
If  $(g,\prec)$ is a maximally delayed causal flow (gflow) of
$(G,I,O)$ then $V^\prec_0=O$. 
\end{lemma}
\begin{proof}
Let $(g,\prec)$ be a maximally delayed causal flow (gflow) of $(G,I,O)$. Elements of $V^\prec_0$ have no image under $g$ because of condition $1$ in both definitions thus $V^\prec_0 \subseteq
O$. Moreover, by contradiction, if $O \setminus V^\prec_0\neq
\emptyset$, let $\prec' =\prec \setminus (O\setminus V^\prec_0)\times
V(G)$.    
$(g,\prec')$ is a causal flow (gflow) of
$(G,I,O)$: condition $1$ of both definitions is satisfied by $\prec'$,
because the domain of $g$ does not intersect $O$, so for any $i$ in the
domain of $g$, $i\prec' j$ iff $i\prec j$; conditions 2 and 3 of both
definitions are satisfied in a same way. 
Thus,  $(g,\prec')$ is a causal flow (gflow) of $(G,I,O)$. Moreover,
for any $k$, $\cup_{i=0\ldots k} V_k^\prec \subseteq \cup_{i=0\ldots
k}V_k^{\prec'}$, and $\vert V_0^\prec\vert < \vert V_0^{\prec'}\vert$
thus $(g,\prec')$ is more delayed than $(g,\prec)$ which leads to a contradiction.
\end{proof}

\begin{lemma}\label{lem:flowv1}
If $(g,\prec)$ is a maximally delayed  causal flow (gflow) of
$(G,I,O)$ then $(\tilde{g},\tilde{\prec})$ is a maximally delayed
causal flow (gflow) of $(G,I,O \cup V^\prec_1)$ where $\tilde{g}$ is the restriction of $g$ to $V(G) \setminus (V^\prec_0 \cup V^\prec_1)$ and $\tilde{\prec}=\prec \setminus V^\prec_1\times V^\prec_0$.
\end{lemma}
\begin{proof}
First, one can prove that $(\tilde{g},\tilde{\prec})$ is a causal flow
(gflow) of $(G,I,O \cup V^\prec_1)$. Moreover, 
by contradiction, if there exists a causal flow (gflow) $(g',\prec')$
that is more delayed than $(\tilde{g},\tilde{\prec})$ then it could be
extended to $(g'',\prec'')$ where $g''(u)=g'(u)$ if $u \in V \setminus
(V^\prec_0 \cup V^\prec_1)$, $g''(u)=g(u)$ if $u \in V^\prec_1$ and
$\prec'' = \prec' \cup \{(u,v), u\in V^\prec_1 \wedge u\prec v\} $. 
$(g'',\prec'')$ is then a more delayed causal flow (gflow) of $(G,I,O)$ than
$(g,\prec)$, which leads to a contradiction.
\end{proof}

\begin{lemma}\label{lem:gflowv1}
If $(g,\prec)$ is a maximally delayed  gflow, then  $V^\prec_1=\{ u\in V \setminus O$, $\exists K \subseteq O$,  $Odd(K) \cap (V\setminus O)=\{ u\}\}$.
\end{lemma}
\begin{proof}
First, notice that if $(g,\prec)$ is a maximally delayed  gflow, then
for any $u\in V^\prec_1$, $g(u)\subseteq O$ since $u\prec v$ if $v \in g(u)$
(condition 1 of definition \ref{def:gflow}.) Furthermore, by
definition of $V^\prec_1$, if $u\prec v$ then $v\in O$ thus conditions 2 and
3 of definition \ref{def:gflow} imply that $Odd(g(u)) \cap (V\setminus O)=\{ u\}$.

To prove that any $u\in V \setminus O$ such that $\exists K \subseteq
O$,  $Odd(K) \cap V\setminus O=\{ u\}$, $u\in V^\prec_1$, we proceed by
contradiction. We prove that delaying the measurement of  a vertex not
in $V^\prec_1$ satisfying the condition permits to create a more delayed
gflow. Indeed, let $(g,\prec)$ be a maximally delayed flow of
$(G,I,O)$ and let $u_1 \in V \setminus V^\prec_0$ be such that
$\exists K \subseteq O$,  $Odd(K) \cap V\setminus O=\{ u_1\}$. Let $g'
(u)= K$ if $u=u_1$ and $g'(u)= g(u)$ otherwise. Let $\prec'$ be the strict
partial order defined by  $u \prec'v$ if $u \neq u_1$ and
$u\prec v$ or if $u= u_1$ and $v \in K$. It leads to a contradiction
since $(g',\prec')$ is a more delayed gflow
of $(G,I,O)$  than $(g,\prec)$.
\end{proof}

In a similar way, one can prove that:
\begin{lemma}\label{lem:causalflowv1}
If $(g,\prec)$ is a maximally delayed causal flow, then $V^\prec_1=\{ u\in V
\setminus O$, $\exists v \in O$, $N(v) \cap V\setminus O=\{ u\}\}$.
\end{lemma}

Lemmas \ref{lem:gflowv1} and \ref{lem:causalflowv1} show that in a
maximally delayed flow, all the elements that can be corrected at the
last step are in the maximal layer of $V\setminus O$ (i.e. in
$V^\prec_1$.) Combined with the recursive structure of maximally delayed
flow (lemma \ref{lem:flowv1}), this shows that the layers $V^\prec_k$ of a maximally
delayed flow can be iteratively constructed by finding   elements that
can be corrected starting from the output qubits. This gives rise to
the polynomial time algorithms of the next sections.

\section {Causal flow algorithm}
The problem of finding a causal flow of a given open graph is presented in \cite{DK}, and a
solution has been proposed in the case where the numbers of
inputs and outputs are the same. The complexity of the algorithm is in
$O(nm)$ where $n$ is the number of vertices and
$m$ the number of
edges (more precisely $O(km)$ where $k$ is the number of inputs
(outputs) \cite{B07b}.) We present here a more general and faster algorithm.

\begin{theorem}\label{thm:algoflow}
For a given open graph $(G,I,O)$, finding a causal flow can be done in $O(k.n +m)$ operations where $n=|V(G)|$ is the number of vertices of $G$, $m=|V(E)|$ is the number of edges and $k=|O|$ is the size of the output. 
\end{theorem}

In order to prove  Theorem \ref{thm:algoflow}, we introduce the algorithm
\ref{algo:flow} which decides whether given an open graph has a causal
flow, and outputs a maximally delayed
causal flow if one exists. This
recursive algorithm is based on the recursive structure, pointed out in the previous
section, of the
maximally delayed causal flows.

The algorithm recursively finds the layers $(V^\prec_k)_{k=0\ldots d^\prec}$: at the $k^{th}$ call
to \verb@Flowaux@, the algorithm finds the set $V^\prec_k=Out'$ (see
algorithm \ref{algo:flow} and figure \ref{fig:algo}.) To improve
the complexity of the algorithm, a set $C$ of potential
correctors ($\forall u \in V^\prec_k, g(u) \in C$) is maintained. The
algorithm produces a subset $C'$ of $C$ of vertices that can be
actually used as correctors, the  set $Out'$ of vertices that can be
corrected by $C'$ is produced as well. 
For the recursive call, the vertices of $Out'$ are added to
the potential correctors, whereas the vertices used as
correctors  (i.e. $C'$)  are removed from the set of potential
correctors since a vertex can be used to correct at most one other
vertex. 

The partial order $\prec$ of the flow found by the algorithm is
defined via a labeling $l$ which associates with each vertex the index
of its layer. As a consequence, for any two vertices $u$ and $v$, $u\prec v$ iff $l(u)>l(v)$.

{\footnotesize    
\incmargin{1em} 
\restylealgo{boxed}\linesnumbered 
\begin{algorithm} \label{algo:flow} \footnotesize 
\SetKwData{In}{In} 
\SetKwData{Out}{Out} 
\SetKwData{Graph}{G}
\SetKwData{correct}{C} 
\SetKwData{Left}{left} 
\SetKwData{This}{this} 
\SetKwData{Up}{up} 
\SetKwFunction{Union}{Union} 
\SetKwFunction{Flow}{Flow} 
\SetKwFunction{Flowaux}{Flowaux} 
\SetKwFunction{FindCompress}{FindCompress} 
\SetKwInOut{Input}{input} 
\SetKwInOut{Output}{output} 
\caption{Causal flow} 
\Input{An open graph} 
\Output{A causal flow}  
\BlankLine 
\Flow(\Graph,I,O) =

\Begin{

\For{{\bf all} $v\in O$ }{$l(v):=0$ \;}
 \Flowaux(\Graph,I,O,O$\setminus$I,1)\;
}
\BlankLine 
\BlankLine 
\Flowaux(\Graph,\In,\Out,\correct,k) = \% precondition: \Out
$\setminus$ \correct cannot be used 
to correct elements of layer $k${\dontprintsemicolon \;}

\Begin{
\Out':=$\emptyset$\;
\correct':=$\emptyset$\;
\For{{\bf all} $v\in \correct$ }{ 
\If{$\exists u \ s.t. \ N(v)\cap (V\setminus \Out)=\{u\}$}{
$g(u):=v$ \;
$l(v):=k$\;
$Out' := Out' \cup \{u\}$\;
$C' :=C' \cup \{v\}$\;
}
}
\eIf{$Out'=\emptyset$}{
\eIf {$Out=V$}{true}{false}}
{
\Flowaux(\Graph,\In,\Out$\cup$ \Out',(\correct$\setminus$ \correct')$\cup$(\Out'$\cap V\setminus $\In),k+1)
}}
\end{algorithm} 
\decmargin{1em} 
}

\begin{figure}\label{fig:algo}
\begin{center}
\includegraphics[width=6cm]{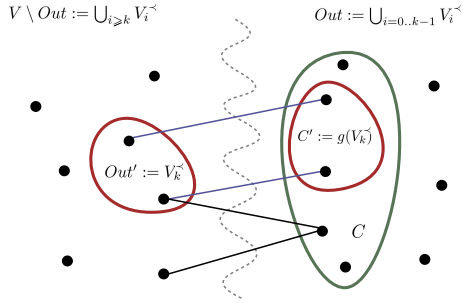}
\end{center}
\caption{\footnotesize { Causal flow algorithm: At the $k^{th}$ recursive call, the
algorithm finds out the set $V^\prec_k$ composed of the qubits that will be
measured at the $d^\prec-k+1$ step of the one-way quantum computation, where
$d^\prec$ is the depth of the computation.   
At that step, all the qubits in
$Out := \bigcup_{i=0..k-1}V^\prec_k$ are not measured, whereas the qubits in
$\bigcup_{i>k}V^\prec_k$ are already measured. The correctors of the elements
of $V^\prec_k$ are in a set $C\subseteq Out$ of candidates composed of
vertices not already assigned to the correction of some future
measurement. 
The first stage of the algorithm to find out the set $V^\prec_k$ consists in
searching, among $C$, for the elements that have a unique neighbor in
$V\setminus Out$. Let $C' \subseteq C$ be this set of correctors. Then,
the neighborhood $Out'$ of $C'$ in $V\setminus Out$ is a set of
elements that can be corrected at that step, so $Out'$ is nothing but
$V^\prec_k$. For the recursive call, the elements of $Out'$ are added to
both $Out$ and $C$, whereas  the elements of $C'$ are removed from
$C$. Since at each step, a maximum number of vertices are added to
$V^\prec_k$, the causal flow, if it exists, produced by this algorithm is
maximally delayed.}}
\end{figure}

\ 

\noindent \textbf{Proof of Theorem \ref{thm:algoflow}:} 

By  induction on the number of non output qubits, we prove  that if
the given open graph has a causal flow then the algorithm outputs a
maximally delayed one. 
Assume that the given open graph has a causal flow. First, if there is no non output qubit, then no correction is needed:
the empty flow $(g,\emptyset)$ (where $g$ is a function with an empty
domain) is a maximally delayed gflow. 
Now suppose that there exist some non output vertices,
 according to lemma \ref{lem:causalflowv1} the elements of $V^\prec_1$ satisfy
the test at line 13, moreover the 
precondition at line 8 (that can reformulated as
$g(V^\prec_1)\subseteq C$) implies that $V_1^\prec$ is composed of the elements that satisfy the
test at line 13. Thus, after the loop (line 19),
$Out'=V^\prec_1$ and $C'=g(V^\prec_1)$. Since the existence of a causal
flow
is assumed, $V^\prec_1$ cannot be empty (all non output qubits have to be corrected), thus the algorithm is called recursively. 
Lemma \ref{lem:flowv1} ensures the existence of a causal  flow in
$(G,I,O \cup V^\prec_1)$  and since the vertices in $C'$ have no neighbor in $V\setminus (O \cup
V^\prec_1)$, they can be removed from the set of potential correctors,
preserving the precondition. 

The induction hypothesis ensures that the recursive calls output a
maximally delayed causal flow in $(G,I,O \cup V^\prec_1)$ and thus the
causal flow $(g,\prec)$ defined is a maximally delayed causal flow of $(G,I,O)$.

The termination of the algorithm is ensured by the fact that the set of output qubits strictly increases at each recursive call.

For a given open graph, if the algorithm outputs a flow $(g,\prec)$, then this
flow is a valid causal flow since every output qubit has an image
under $g$, moreover for any vertex $i$, $i\prec g(i)$, and finally if
$j\in N(g(i))$ then $j=i$ or $i\prec j$. Thus, if the given open graph
has no flow, the algorithm returns false.

To analyze the complexity of the algorithm, we consider the cost for each vertex $u$, which can be decomposed in:
\begin{itemize}
\item Insert $u$ in the set $C$ of potential correctors
\item Create  the set of vertices that $u$ might correct  $N(u) \cap (V \setminus Out)$
\item Check whether $u$ can be removed  $|N(u) \cap (V \setminus Out)| =1$
\item Update the potential correctors sets $N(v) \cap (V \setminus Out)$ for $v\in C$  when $u$ is removed ($u$ belongs to $Out$ for the recursive call).
\end{itemize}

As there is at most $|Out|$ potential correctors, the cost for a vertex $u$ can be decomposed in: insert in $C$ + create + $|Out|$(remove + check).

Using a data structure for storing the sets   $N(u) \cap (V \setminus Out)$ (for example an array with two pointers respectively to next and previous elements), one can remove an element in constant time and test whether the set contains exactly one element in constant time. For the creation of the structure, one needs to compute the intersection of the neighborhood with  $(V \setminus Out)$. Checking whether a vertex is in $(V \setminus Out)$ can be done in constant time by  maintaining an array of the vertices that are to be corrected, thus given the adjacency list of a vertex $u$ the cost of creating  $N(u) \cap (V \setminus Out)$  is the degree of $u$ and the total finding cost is $O(m)$ where $m$ is the number of edges of the graph. The overall complexity is then $O(n |Out| +m)$. $\hfill \Box$

\ 

 This result improves the algorithm in
\cite{B07} that decides, under the precondition $|I|=|O|$, whether an
open graph $(G,I,O)$ has a causal flow in
$O(km)$ operations, where $k=|O|$. 
In \cite{B07c}, Pei and de Beaudrap have proved that an
open graph having a causal flow has at most $(n-1)k
-\binom{k}{2}$ edges. According to this result,
the algorithm in \cite{B07} can be transformed (see \cite{B07c}) into a $O(k^2n
)$-algorithm, whereas our algorithm becomes a
$O(kn)$-algorithm. 

\section{A polynomial algorithm for gflow}

\begin{theorem}

There exists a polynomial time algorithm that decides whether a given
open graph has a gflow, and that outputs a gflow if it exists.
\end{theorem}
\begin{proof}
Let $(G,I,O)$ be an open graph. The algorithm \verb@gFlow@$(\Gamma,I,O)$ (Algorithm \ref{algo:gflow}), where $\Gamma$ is the adjacency matrix of $G$, finds a maximally delayed gflow and returns true if one exists and returns false otherwise. 
Given a set $Out'$ and a subset $X \subseteq Out'$, $\mathbb I_{X}$
stands for a $|Out'|$-dimensional vector defined by $\mathbb I_{X}(i)=1$ if $i  \in X$ and $\mathbb I_{X}(i)=0$ otherwise.

\incmargin{1em} 
\restylealgo{boxed}\linesnumbered 
\begin{algorithm} \label{algo:gflow}\footnotesize
\SetKwData{In}{In} 
\SetKwData{Out}{Out} 
\SetKwData{Graph}{G}
\SetKwData{correct}{C} 
\SetKwData{Left}{left} 
\SetKwData{This}{this} 
\SetKwData{Up}{up} 
\SetKwFunction{Union}{Union} 
\SetKwFunction{gFlow}{gFlow} 
\SetKwFunction{gFlowaux}{gFlowaux} 
\SetKwFunction{FindCompress}{FindCompress} 
\SetKwInOut{Input}{input} 
\SetKwInOut{Output}{output} 
\caption{Generalized flow} 
\Input{An open graph} 
\Output{A generalized flow} 
\BlankLine 
\gFlow($\Gamma$,\In,\Out) =

\Begin{
\For{{\bf all} $v\in \Out$ }{ $l(v):=0$ \;}
 \gFlowaux($\Gamma$,\In,\Out,\Out$\setminus$ \In,1)\;
}
\BlankLine 
\BlankLine 

\gFlowaux($\Gamma$,\In,\Out,$k$) = {\dontprintsemicolon \;}
\Begin{
\Out':= \Out $\setminus$ \In\;
\correct:=$\emptyset$\;
\For{{\bf all} $u\in  V\setminus$ \Out }{ 
Solve in $\mathbb F_2 : \Gamma_{V\setminus \text{\Out,\Out'}} \mathbb I_X = \mathbb I _{\{u\}}$ \;
If there is a solution $X_{0}$ then \correct:=\correct$\cup \{u\}$ and $g(u):=X_{0}$ \;
$l(u)=k$\;
}
\eIf{\correct = $\emptyset$}{ \eIf{Out=V}{true}{false}}
{\gFlowaux($\Gamma$,\In,\Out$\cup$ \correct,k+1)}}

\end{algorithm} 
\decmargin{1em}

At the $k^{th}$ recursive call, the set $C$ found by the algorithm at
the end of the loop at line 16 corresponds to  the layer $V^\prec_k$
of the partition induced by the returned strict partial order. 
At line 13, the columns of the matrix $\Gamma_{V\setminus Out,Out'}$
correspond to the vertices that can be used for correction (vertices
in $\cup_{i<k}V^\prec_k \setminus In$) and the rows to the candidates for
the set $V^\prec_k$. A solution $X_0$ in $\mathbb F_2$ to
$\Gamma_{V\setminus Out,Out'}\mathbb I_{X} = \mathbb I_{\{u\}}$
corresponds to a subset of $\cup_{i<k}V^\prec_k \setminus In$ that has only
$u$ as odd neighborhood in $\cup_{i\le k} V^\prec_i$, thus $g(u):=X_0$
satisfies  conditions 2 and 3 required by the  definition of  gflow
(see Definition \ref{def:gflow}). Furthemore, line 10 of the algorithm ensures  condition 1, thus
if the algorithm returns a flow, then it  satisfies the definition of  gflows.

Now suppose that the  graph admits a gflow $(g,\prec)$, then it also admits a maximally delayed gflow $(g',\prec')$. The algorithm finds the set $V'_1$ (in the loop at line 12), and by induction (similarly to the induction in the proof of Theorem \ref{thm:algoflow}) it also finds a maximally delayed  gflow in $(G,I,O \cup V^\prec_1)$ with the recursive call. Thus the algorithm finds  a maximally delayed gflow.
\end{proof}

In order to analyse the complexity, notice that lines 10 to 16
consists in solving a system $Ax=b_i$ for $n-\ell$ different $b_i$s
where $n=\vert V\vert$, $\ell=\vert Out\vert $ and $A$ is a $(n-\ell)\times
\ell$ matrix. In order to solve these $n-\ell$ systems, the $(n-\ell)\times n$-matrix
$M=[A\vert b_1\ldots b_{n-\ell}]$ is transformed into an upper
triangular form within $O(n^3)$ operations using gaussian
eliminations for instance, then for each $b_i$ a back substitution within $O(n^2)$
operations is used to find $x_i$, if it exists, such that
$Ax_i=b_i$ (see \cite{B06}). The back substitutions costs $O(n^3)$ operations at each
call of the function. Since there are at most $n$ recursive calls, the
overall complexity is $O(n^4)$.

\section{Depth Optimality}
We consider in this section the depth of the flows found by the algorithms, which corresponds to the number of steps required by the correction strategy, and we show that both algorithms find minimal depth flows, hence optimal correction strategies.
\begin{theorem}
\label{opt}
The previous algorithms find an optimal depth flow
\end{theorem}
\begin{proof}
First, notice that if $(g,\prec)$ is more delayed then $(g',\prec')$ then
$|\cup_{i=0\ldots d^{\prec'}} V^\prec_{k}|  \ge   |\cup_{i=0\ldots
d^{\prec'}} V^{\prec'}_{k}|=|V|$. Thus $\forall k> d^{\prec'}$,
$V^{\prec}_k=\emptyset$, so $d^\prec \le d^{\prec'}$. 

Given an open graph $(G,I,O)$, and an optimal depth flow $(g,\prec)$, we can define a maximally delayed optimal depth flow $(g',\prec')$:  If $(g,\prec)$ is maximally delayed then $(g',\prec')=(g,\prec)$ otherwise $(g',\prec')$ is a  maximally delayed flow that is more delayed than $(g,\prec)$. According to the previous remark $(g',\prec')$ is of optimal depth as well and the  optimal depth $d(G,I,O)=d^{\prec} =d^{\prec'}$.
 By lemma  \ref{lem:flowv0}, 
$V_1^{\prec'}=\{u,\exists K \subseteq O$,  $Odd(K) \cap V\setminus O=\{ u\}\}=V_1^{\prec''}$ where $(g'',\prec'')$ is the flow found by the algorithm. Thus $d(G,I,O)=1+d(G,I,O \cup
V^{\prec''})$.  By
induction on the number of output vertices the algorithm gives an optimal flow for  $(G,I,O \cup
V^{\prec''})$   so the depth given by the algorithm is the depth of
$(g',\prec')$  which is optimal.

\end{proof}

The optimality of the previous algorithms have several implications in one-way
quantum computation. First, the depth (optimal or not) of a flow is an
upper bound on the depth of the corresponding deterministic one-way
quantum computation. 
Moreover, if the one-way quantum computation is uniformly, stepwise and strongly
deterministic (which mainly means that if the measurements are applied
with an error in the angle which characterises the measurement,
then the computation is still deterministic), then the correction
strategy must be described by
a gflow \cite{BKMP}. As a consequence the algorithm 2 produces the optimal
correction strategy, and the depth of the gflow produced by the
algorithm is a lower bound on the depth of a uniformly, stepwise and
strongly deterministic one-way quantum computation.

\section{Conclusion}
Starting from quantum computational problems (determinism in one-way
quanum computation), interesting graph problems have arisen like the
graph flow dependence of the depth of correcting strategies for measurement-based quantum computation. 

We have defined in this paper two algorithms for finding optimal
causal flow and gflow. The key points are: the simplification of the structure of the flows considering only the maximally delayed flows
which  have a nice recursive structure;  a backward analysis (start from the outputs) which allows to take advantage of this structure and avoids backtracking.

From a complexity point of view, an important  question is:  given a graph state and a fixed set of measurements (we relax the uniformity condition) what would be the depth of an optimal correction strategy. One direction to answer this question would be to define a weaker flow that is still polynomially computable.

One can also consider the characterization and the depth of computation in more generalized measurement-based models where  other planes of measurements are allowed.

\section*{Acknowledgements}
The authors would like to thank Elham Kashefi and Philippe Jorrand for
fruitful discussions.

\end{document}